\documentclass{article}

\usepackage[english]{babel}

\usepackage[letterpaper,top=2cm,bottom=2cm,left=3cm,right=3cm,marginparwidth=1.75cm]{geometry}

\usepackage{amsmath}
\usepackage{graphicx}
\usepackage[colorlinks=true, allcolors=blue]{hyperref}
\usepackage{amsthm}
\usepackage{tikz}
\usepackage{centernot}
\usepackage{mathtools}
\usepackage{ stmaryrd }
\usetikzlibrary{matrix,arrows.meta}
\usepackage{amsmath}
\usepackage{tikz}

\usepackage{amsmath}
\usepackage{amssymb}
\usepackage{amsfonts}
\usepackage{amsthm,amscd}
\usepackage{xcolor}
\newtheorem{theorem}{Theorem}[section]
\newtheorem{lemma}[theorem]{Lemma}

\newtheorem{proposition}[theorem]{Proposition}
\theoremstyle{definition}
\newtheorem{definition}[theorem]{Definition}

\usepackage{xcolor}
\usepackage[dvipsnames]{xcolor}
\theoremstyle{remark}
\newtheorem{remark}[theorem]{Remark}

\usetikzlibrary{matrix,arrows,decorations.pathmorphing}
\usetikzlibrary{shapes,arrows,positioning}
\usepackage{graphicx}
\usepackage{amsmath}
\usepackage{amssymb}
\usepackage{latexsym}
\usepackage{amssymb,amsmath}
\usepackage{tikz-cd}
\usetikzlibrary{arrows}
\usepackage{amsfonts}
\usepackage{amssymb}
\usepackage{graphicx}
\usepackage{fancyhdr}
\usepackage{fancyvrb}
\usepackage{fancybox}
\usepackage{float}
\usepackage{geometry}
\usepackage{minitoc}
\usepackage{indentfirst}
\usepackage{multicol, color}
\usepackage[outerbars]{changebar}
\usepackage{pifont,textcomp}
\usepackage{amsfonts}
\usepackage{amsmath}
\usepackage{amscd}
\usepackage{array}
\usepackage{multirow}
\usepackage{mathrsfs}
\usepackage{amssymb}
\usepackage{amsthm}

\usepackage{amsmath}
\usepackage{graphicx}
\usepackage[colorlinks=true, allcolors=blue]{hyperref}

\title{\bfseries\scshape{Deformed Heisenberg algebra and its Hilbert space representations}}
\title{\bfseries\scshape{Deformed Heisenberg algebra and its Hilbert space representations}}
\author{\bfseries\scshape Latévi M. Lawson\thanks{E-mail address: \tt{latevi@aims.edu.gh}\,(Corresponding author)}\\
 African Institute for Mathematical Sciences (AIMS), \\1 Shoppers Street, Manet, Spintex, Accra, Ghana.\\ 
\\\bfseries\scshape Ibrahim Nonkané\thanks{E-mail address: \tt{ inonkane@univouaga2.bf}}\\
Département d’Economie et de Mathematiques Appliquées, ´
\\IUFIC, Universite Thomas Sankara, Burkina faso \\
\\\bfseries\scshape Kinvi Kangni\thanks{E-mail address: \tt{ kinvi.kangni@univ-fhb.edu.ci}}\\
  UFR de Mathématiques et Informatique, Université Felix Houphouët Boigny-Abidjan\\ 22 BP 582 Abidjan, 22 Côte d’Ivoire}
\begin{document}
\maketitle
\begin{abstract}
    A deformation of   Heisenberg algebra induces among other consequences  a loss of Hermiticity of some operators that generate this algebra.  Therefore, these operators are not Hermitian, nor is the Hamiltonian operator built from them. In the present paper, we propose a position deformation of  Heisenberg algebra with both maximal
length and minimal momentum uncertainties.  By using a pseudo-similarity transformation to the non-Hermitian operators, we prove their Hermiticity with a suitable positive-definite pseudo-metric operator. We then construct Hilbert space representations associated with these pseudo-Hermitian operators. Finally, we study the eigenvalue problem of a free particle and an ordinary harmonic oscillator   in this deformed  space.
\end{abstract}

\section{Introduction}
Traditional quantum mechanics relies on Hermitian
Hamiltonians to ensure real eigenvalues and unitary time
evolution with respect to the inner product, resulting in conservation of the inner product under this time evolution. However, there are exception to this rule. Hermiticity is not a necessary condition for eigenvalues to be real i.e., the reality of eigenvalues does not require that the Hamiltonian be Hermitian \cite{1, 2,3,4,5,6,7,8,9,10,11,12,13,14}. A familiar example of a matrix Hamiltonian which is not Hermitian and yet has real eigenvalues is given by
\begin{eqnarray}
    \hat H=\begin{pmatrix}
\alpha & \beta \\
0 & \gamma 
\end{pmatrix}\neq \hat H^\dag,\quad\quad \alpha,\beta,\gamma\in  \mathbb{R},
\end{eqnarray}
which is not Hermitian, but the eigenvalues of the Hamiltonian are given by $ E=\alpha,\gamma$ which are real. In this case, we note that if we define
\begin{eqnarray} \label{q21}
    S=\begin{pmatrix}
a & -\frac{\beta}{\alpha-\beta}a \\
 \frac{\beta}{\alpha-\beta}a & d 
\end{pmatrix}, \quad  S^{-1}=\frac{1}{det S}\begin{pmatrix}
d & -\frac{\beta}{\alpha-\beta}a \\
- \frac{\beta}{\alpha-\beta}a & a
\end{pmatrix},
\end{eqnarray}
where $a, d$ are arbitrary constants, we can check that
\begin{eqnarray}\label{q11}
    \hat H= S^{-1}\hat H^\dag S.
\end{eqnarray}
Any quantum mechanical system whose Hamiltonian can be related to its adjoint  through a similarity transformation as in \eqref{q11} is called a pseudo-Hermitian Hamiltonian \cite{1,2,3,4,5,6,7}.

 Our focus in this study is on non-Hermitian Hamiltonian systems with real spectra \cite{1,2,3,4,5,6,7,8,9,10,11,12,13,14,15,16,17,18,19,20,21,22,24,25,26,27,28} which  are consistently described by   pseudo-Hermitian quantum mechanical approaches. In this  context, we are interested in a position deformation of  Heisenberg algebra with  both maximal length and minimal momentum uncertainties \cite{23,25,26,27}. 
 As has been shown in  \cite{29}, this deformation induces a loss of Hermiticity of the
momentum operator which consequently forms with the position operator a non-Hermitian position deformed Heisenberg algebra. The Hamiltonian operator that incorporates this non-Hermitian operator is therefore not a Hermitian operator either, and the inner product is not conserved under the time-evolution.

 Consequently, Hamiltonian operator that includes this non-Hermitian operator is not a Hermitian operator as well and non-conservation of the inner product under the time-evolution.  To raise up the Hermiticity issue of this operator, we propose an appropriate positive-definite pseudo-metric operator \eqref{q21} and by means of a similarity transformation \eqref{q11},  we relate the non-Hermitian Hamiltonian   to its adjoint . We construct the position wave function and its Fourier transform that describes the corresponding Hilbert space within this deformed Heisenberg algebra.  Finally, we study the eigenvalue problem of a free particle and an ordinary harmonic oscillator   in this deformed  space.

This paper is outlined as follows: In section \eqref{sec2}, we review fundamentals of the pseudo-hermiticity  quantum mechanics. In section \eqref{sec3}, we apply the concept of pseudo-hermiticity to
the position and momentum operators that enter the quadratic position-deformed Heisenberg algebra.
In section \eqref{sec4}, we construct Hilbert space representations associated with this deformed algebra. Section \eqref{sec5} provides  an application of an eigenvalue problem to the Hamiltonians of a free particle  and a harmonic oscillator   in this deformed Heisenberg algebra.
 In the last section \eqref{sec6}, we present our conclusion.

\section{Pseudo-Hermiticity in quantum mechanics}\label{sec2}

\begin{definition} {\it   Let $\mathcal{H}$ be a finite dimensional Hilbert space  with the standard scalar product  $ \langle . |. \rangle.$ A  non-Hermitian Hamiltonian $\hat H:\mathcal{H}\rightarrow \mathcal{H}$  is
	said to be  pseudo-Hermitian \cite{9}, if there exists an automorphism ${S}:\mathcal{H}\rightarrow \mathcal{H}$ i.e.  an invertible, Hermitian, linear operator  satisfying
	\begin{eqnarray}\label{aaa1223}
		\hat H^\dag = {S}\hat H{S}^{-1}.
\end{eqnarray} } 
\end{definition}
Being an automorphism, its domain of definition
is the entire space, so that  by virtue of the theorem of Hellinger and Toeplitz, it is bounded.  To build a Hilbert space, one must satisfy the criterion that ${S}$ must be positive-definite denoted by  $S_+$ since it guarantees the positive-definiteness of the scalar product. 

\begin{definition}{\it A  Hilbert space $\mathcal{H}^{S_+}$ endowed  with a new  inner product $ \langle . |. \rangle_{S_+}$ in terms of the standard inner product $ \langle . |. \rangle $ is defined by
	\begin{eqnarray}\label{S11}
		\langle \psi|\phi \rangle_{S_+}:= \langle \psi| {S_+}\phi \rangle= \langle S_+^\dag\psi|\phi \rangle, \quad  \quad \psi,\phi\in \mathcal{H}.
\end{eqnarray} }
\end{definition}
For brevity we shall call the latter a pseudo-inner
product. Since the operator $S_+$ is positive-definite, one can easily show that $ \langle .|. \rangle_{S_+}$  is  positive-definite, non-degenerate and Hermitian \cite{30}. With the boundedness of $S_+$ one can show that $\mathcal{H}^{S_+}$   forms a  complete space \cite{9}  with the norm $ ||\phi ||_{S_+}=\sqrt{\langle \phi|\phi \rangle_{S_+}}$. In this way, the scalar product $ \langle .|. \rangle_{S_+}  $ can serve as the basis of a quantum theory.

Note that, the pseudo-Hermiticity of an operator is not sensitive to the particular form $S_+$ of the operators  satisfying $ \hat H^\dag = {S}_+\hat H{S}_+^{-1}$  but to the existence of such operators. However, for a fixed operator ${S}_+$, the linear operator  $\hat H:\mathcal{H}\rightarrow \mathcal{H}$ satisfying \cite{9,13,20,31,32}  \eqref{aaa1223} is called ${S}_+$-pseudo-Hermitian, ${S}_+$-pseudo-Hermitian operators are pseudo-Hermitian, but not every pseudo-Hermitian operator is ${S}_+$-pseudo-Hermitian. This is because, ${S}_+$ may not be defined on the entire space $\mathcal{H}$. 

Moreover, one can ensure the conservation of the conventional probability interpretation of quantum mechanics with the use of this new inner product \eqref{S11}. To do this, we shall demonstrate that, relative to this inner product, the operator Hamiltonian is Hermitian.

\begin{proposition} {\it A non-Hermitian operator $\hat H $ is Hermitian with respect to the pseudo-inner product $ \langle .|. \rangle_{{S}_+}$ if we have
	\begin{eqnarray}
		\langle \psi|\hat H\phi \rangle_{{S}_+}:= \langle  \psi|{S}_+ \hat H \phi \rangle = \langle  \psi| \hat H^{\dag} {{S}_+}\phi \rangle = \langle \hat H \psi|{S}_+\phi \rangle=\langle S_+^{\dag}\hat H \psi|\phi \rangle=\langle \hat H \psi|\phi \rangle_{{S}_+}.
\end{eqnarray}}
\end{proposition}
Operators, such as  $\hat H$, which are Hermitian under the  pseudo-inner product $ \langle .|. \rangle_{{S}_+}$ are called ${S}_+$-pseudo-Hermitian operators \cite{9,32}.

 \begin{lemma} 
 	Since the  Hamiltonian is Hermitian with respect to the inner product $  \langle .|. \rangle_{S_+}$,   this will result in conservation of probability under time evolution 
 \begin{eqnarray}
 	\langle \psi(t)|\phi(t)\rangle_{{S}_+}&=&	\langle \psi(t)|S|\phi(t)\rangle=\langle \psi(0)| e^{\frac{i}{\hbar}t\hat H^\dag} {{S}_+}e^{-\frac{i}{\hbar}t\hat H} |\phi(0)\rangle \cr
 	&=& \langle \psi(0)|{{S}_+}\left({{S}_+}^{-1} e^{\frac{i}{\hbar}t\hat H^\dag}{{S}_+}\right) e^{-\frac{i}{\hbar}t\hat H} |\phi(0)\rangle\cr
 	&=& \langle \psi(0)|{{S}_+} e^{\frac{i}{\hbar}t\hat H} e^{-\frac{i}{\hbar}t\hat H} |\phi(0)\rangle\cr
 	&=& 	\langle \psi(0)|\phi(0)\rangle_{{S}_+}.
 \end{eqnarray}
 \end{lemma}
 
 As we observe, the pseudo-Hermiticity ensures  that the time evolution operator  $e^{-\frac{i}{\hbar}t \hat H}$  is unitary   with respect to this  inner product. 
 
 \section{ Position-deformed Heisenberg algebra }\label{sec3}
 
 Let $\hat x_0=\hat x_0^\dag$ and $\hat p_0=\hat p_0^\dag$ be respectively Hermitian position and momentum operators defined as follows
 \begin{eqnarray} \label{a1}
 	\hat x_0\phi(x)=x\phi(x) \quad \mbox{and}\quad  \hat p_0 \phi(x)=-i\hbar\frac{d}{dx}\phi(x),\label{3}
 \end{eqnarray}
 where $\phi(x)\in \mathcal{H}= \mathcal{L}^2\left(\mathbb{R}\right)$ is the one  infinite  dimensional (1D)    Hilbert  space.  Hermitian  operators $\hat x_0$ and  $\hat p_0$  that act in $\mathcal{H}$  satisfy the  Heisenberg algebra
 \begin{eqnarray}\label{alg1}
 	{[\hat x_0,\hat p_0 ]}=i\hbar\mathbb{I}\quad \mbox{and}\quad 	{[\hat x_0,\hat x_0 ]}=0=	{[\hat p_0,\hat p_0 ]}.
 \end{eqnarray}
 The  Heisenberg uncertainty principle reads as
 \begin{eqnarray}
 	\Delta x_0\Delta p_0\geq \frac{1}{2} \big| \large\langle\phi|[\hat x_0,\hat p_0]|\phi \large\rangle \big|\implies 
 	\Delta x_0\Delta p_0\geq\frac{\hbar}{2}.
 \end{eqnarray}

 Let    $\mathcal{H}_\tau=\mathcal{L}^2\left(\Omega_\tau\right)$ be  a finite dimensional subset of  $\mathcal{H}$ such  that $\Omega_\tau\subset \mathbb{R}$ and $\tau\in (0,1)$ is a deformation
 parameter.   Let  $\hat X$ and $\hat P$ be respectively  position and  deformed momentum operators defined in $\mathcal{H}_\tau$ such that
 \begin{eqnarray}\label{a241}
 	\hat X=\hat x_0\quad \mbox{and}\quad \hat P= \left(\mathbb{I}-\tau \hat x_0+\tau^2\hat x_0^2\right)\hat p_0.
 \end{eqnarray}
 These  operators (\ref{a241})  form the following  position-deformed Heisenberg algebra \cite{23,25,26,27}
 \begin{eqnarray}\label{alg41}
 	[\hat X,\hat P]=i\hbar \left(\mathbb{I}-\tau \hat X+\tau^2\hat X^2\right), \quad 	{[\hat X,\hat X ]}=0=	{[\hat P,\hat P ]}.
 \end{eqnarray}
 From  the representation (\ref{a241}), it follows immediately that  the operator $\hat X$ is Hermitian while the operator $\hat P$ is  no longer Hermitian on the space $\mathcal{H}_\tau$ 
 \begin{eqnarray}\label{ra}
 	\hat X^\dag=\hat X\quad \mbox{and}\quad \hat P^\dag=\hat P-i\hbar\tau(\mathbb{I}-2\tau\hat X)\implies \hat P^\dag\neq \hat P, 
 \end{eqnarray}
 and when $\tau\rightarrow 0$, the momentum operator  $\hat P$  becomes Hermitian. The non-Hermiticity of the momentum operator $\hat P$ is induced by the deformation parameter $\tau$.
  This algebra is hence designated as a non-Hermitian position-deformed Heisenberg algebra. Furthermore, a Hamiltonian operator that includes this non-Hermitian operator  in  representation \eqref{a241}, is not a Hermitian operator as well and   nonconservation of the inner product under the time evolution
 $\langle \psi(t)|\phi(t)\rangle\neq \langle \psi(0)|\phi(0)\rangle, \quad |\psi\rangle ,|\phi\rangle \in \mathcal{H}_\tau$.

 In order to map 
 these operators (\ref{ra}) into the  pseudo-Hermitian ones,  we propose the metric operator ${S}_+$ given by 
 \begin{eqnarray}\label{b1}
 	S_+= \left(\mathbb{I}-\tau \hat X+\tau^2\hat X^2\right)^{-1}.
 \end{eqnarray}
 It is easy to see that  the operator ${S}_+   $ is positive-definite  (${S}_+ >0$), Hermitian (${S}_+={S}_+^\dag$), and invertible. Since $\mathcal{H}_\tau$  is finite dimensional,  $ S_+ $ is bounded.     The pseudo-Hermicities are obtained by means of pseudo-similarity transformation 
 \begin{eqnarray}
 	{S}_+ \hat X {S}_+^{-1}&=& \hat x_0= \hat X^\dag, \label{a21}\\
 	{S}_+ \hat P  {S}_+^{-1}&=& \hat p_0\left(\mathbb{I}-\tau \hat x_0+\tau^2\hat x_0^2\right)= \hat P^\dag.\label{a22}
 \end{eqnarray}
 Using equations (\ref{a21}) and (\ref{a22}), we obtain the pseudo-Hermiticity of the Hamiltonian $\hat H$ such that
 \begin{eqnarray}\label{her}
 	{S}_+ \hat H S_+^{-1}= \frac{1}{2m}\hat p_0\left(\mathbb{I}-\tau \hat x_0+\tau^2\hat x_0^2\right)\hat p_0\left(\mathbb{I}-\tau \hat x_0+\tau^2\hat x_0^2\right)+V(\hat x_0)= \hat H^\dag.
 \end{eqnarray}
 A  Hilbert space $\mathcal{H}_\tau^{{S}_+}$ endowed  with a new  inner product $ \langle . |. \rangle_{{S}_+}$ in terms of the standard inner product $ \langle . |. \rangle $ is defined by
 \begin{eqnarray}\label{S1}
 	\langle \psi|\phi \rangle_{{S}_+}&=& 
 	\langle \psi|{S}_+\phi \rangle= \int_{\Omega_\tau}dx \psi^*(x)\left(\frac{\phi(x)}{1-\tau  x+\tau^2 x^2}\right)=\int_{\Omega_\tau}dx \left(\frac{\psi(x)}{1-\tau  x+\tau^2 x^2}\right)^*\phi(x)\cr&=& \langle S_+^\dag \psi|\phi \rangle.
 \end{eqnarray}
 With the corresponding norm given by
 \begin{eqnarray}\label{a2}
 	|| \phi||_{S_+}= \left(\int_{\Omega_\tau}\frac{dx}{1-\tau  x+\tau^2 x^2} |\phi(x)|^2\right)^{\frac{1}{2}}<\infty.
 \end{eqnarray} 
 	With equation \eqref{S1} at hand, we can demonstrate that the  momentum operator $\hat P$ described in \eqref{a241} and the associated Hamiltonian  are Hermitian \cite{33} with regard to the pseudo-inner product $ \langle.|. \rangle_{S_+}$.
 	\begin{eqnarray}
 		\langle \psi|\hat P\phi \rangle_{{S}_+}= \langle \hat P \psi|\phi \rangle_{{S}_+}\implies \langle \psi|\hat H\phi \rangle_{{S}_+}= \langle \hat H \psi|\phi \rangle_{{S}_+}.
 \end{eqnarray}

 For a system of operators satisfying the commutation relation in (\ref{alg41}), the generalized uncertainty principle
 is defined as follows
 \begin{eqnarray}\label{leq11}
 	\Delta  X\Delta P\geq \frac{\hbar}{2}\left(1-\tau\langle \hat X\rangle_{S_+}+\tau^2\langle \hat X^2\rangle_{S_+}\right), 
 \end{eqnarray}
 where $\langle \hat X\rangle_{S_+}$ and $\langle \hat X^2\rangle_{S_+}$ are the  expectation values of the operators $\hat X$ and $\hat X^2$ respectively for any  space representations.
 Referring to \cite{23,25,27,34,35}, this equation leads to the absolute minimal uncertainty $\Delta P_{min}$ in $P$-direction  and the absolute maximal uncertainty  $\Delta X_{max}$ in $x$-direction  
 when $\langle  \hat X\rangle_{S_+}=0$  such that
 \begin{eqnarray}\label{q2}
 	\quad \Delta X_{max}=\frac{1}{\tau}=\ell_{max}\quad  \mbox{and}\quad
 	\Delta P_{min}=\hbar\tau.
 \end{eqnarray}
 This provides the  scale for the maximum length and minimum momentum obtained in \cite{23,25,27,34,35}  which are
 different from the  condition imposed in  \cite{33}. As we shall see,  in contrast to the earlier conclusion in \cite{33}, the use of these uncertainty values in the current study has no impact on the physical interpretation.

 \section{Hilbert space representations}\label{sec4}
 Let $\mathcal{H}_\tau=\mathcal{L}^2(\Omega_\tau)=\mathcal{L}^2(-\ell_{max},+\ell_{max})\subset \mathcal{H}$ be the Hilbert space. We construct in this section the position space representation and its Fourier transform representation that  describe simultaneously the maximal length  and the minimal momentum uncertainties.

 \subsection{Position space representations }

 {\bf Definition 3:}  {\it Let us consider $ \mathcal{H}_\tau=\mathcal{L}^2\left(-\ell_{max},+\ell_{max}\right)$.
 	The actions of pseudo-Hermitian operators $\hat X,\hat P$  in  $\mathcal{H}_\tau$ read as follows
 	\begin{eqnarray}\label{r1}
 		\hat X\phi(x)=x \phi(x)\quad \mbox{and}\quad	\hat P\phi(x)=-i\hbar D_x\phi(x),
 	\end{eqnarray}
 	where $\phi(x)\in \mathcal{H}_\tau$ and $D_x=\left(1-\tau  x+\tau^2 x^2\right)\partial_x$ is the position-deformed derivation. Obviously, for $\tau\rightarrow 0$, we recover the ordinary  derivation. }\\
  To construct a Hilbert
 space representation that describes the maximal length uncertainty  and the minimal momentum uncertainty (\ref{q2}), one has to solve the following eigenvalue problem
 on the position space
 \begin{eqnarray}\label{a20}
 	-i\hbar D_x \phi_\xi(x)=\xi\phi_\xi(x),\quad\quad \xi\in\mathbb{R}.
 \end{eqnarray}\label{a7}
The solution of this  equation is given by 
\begin{eqnarray}\label{fzeta}
	\phi_\xi (x)= A\exp\left(i\frac{2\xi}{\tau\hbar \sqrt{3}}\left[\arctan\left(\frac{2\tau x-1}{\sqrt{3}}\right)
	+\frac{\pi}{6}\right]\right),
\end{eqnarray}
where $A$ is an abritrary constant.
Then by normalization, $\langle \phi_\xi|\phi_\xi\rangle=1$, we have 
\begin{eqnarray}\label{nzeta}
	A= \left(\int_{-\ell_{max}}^{\ell_{max}}\frac{dx}{1-\tau x+\tau^2x^2}\right)^{-\frac{1}{2}}  =  \sqrt{\frac{\tau\sqrt{3}}{\pi}} \label{a8}.
\end{eqnarray}
Substituting this equation (\ref{nzeta}) into the equation (\ref{fzeta}) gives 
\begin{eqnarray} \label{nor}
	\phi_\xi (x) &=& \sqrt{\frac{\tau\sqrt{3}}{\pi}}  \exp\left(i\frac{2\xi}{\tau\hbar \sqrt{3}}\left[\arctan\left(\frac{2\tau x-1}{\sqrt{3}}\right)
	+\frac{\pi}{6}\right]\right).
\end{eqnarray}
This  wave function describes simultaneously the maximal length  and the minimal momentum uncertainties.  As one  can see, the eigenvectors $|\phi_\xi\rangle$ are  physical states. This because, they are square integrable wavefunction such that
\begin{eqnarray}\label{az1}
	||\phi_\xi||_{S_+}^2=  \frac{\tau\sqrt{3}}{\pi}    \int_{-\ell_{max}}^{+\ell_{max}} \frac{dx}{1-\tau x+\tau^2 x^2}<\infty
\end{eqnarray}
and  the  expectation values of position energy operator  $\hat X^n$ ($n\in \mathbb{N}$) are finite.
\begin{eqnarray}
	\langle \phi_\xi| \hat X^n| \phi_\xi\rangle_{S_+}=   \int_{-\ell_{max}}^{+\ell_{max}}\frac{x^n dx}{1-\tau x+\tau^2 x^2}<\infty.
\end{eqnarray}
One can also show that these states are non-orthogonal such as
\begin{eqnarray}
	\langle \phi_{\xi'}|\phi_{\xi}\rangle_{S_+}&=& \frac{\tau\sqrt{3}}{\pi}\int_{-\ell_{max}}^{+\ell_{max}}\frac{dx}{1-\tau x+\tau^2 x^2}\exp\left(i\frac{2(\xi-\xi')}{\tau\hbar \sqrt{3}}\left[\arctan\left(\frac{2\tau x-1}{\sqrt{3}}\right)+\frac{\pi}{6}\right]\right)\cr
	&=& \frac{\tau\hbar\sqrt{3}}{\pi(\xi-\xi')}\sin\left(\pi\frac{\xi-\xi'}{\tau\hbar \sqrt{3}}\right).\label{za}
\end{eqnarray}
This relation shows that, the normalized eigenstates \eqref{fzeta} are no longer orthogonal. In addition of the latter  physically properties of the states  $|\phi_\xi\rangle$, on can show that these exhibit a Gaussian distribution. So, if one tends $(\xi - \xi')\rightarrow \infty $, these states become orthogonal
\begin{eqnarray}\label{a10}
	\lim_{(\xi - \xi')\rightarrow \infty} \langle \phi_{\xi'}|\phi_{\xi}\rangle_{S_+}=0.
\end{eqnarray} 
However, for  $(\xi - \xi')\rightarrow 0,$ we have 
\begin{eqnarray}\label{aa11}
	\lim_{(\xi - \xi')\rightarrow 0} \langle \phi_{\xi'}|\phi_{\xi}\rangle_{S_+}=1.
\end{eqnarray}
These properties show that, the states $|\phi_{\xi}\rangle$ are essentially Gaussians centered at $(\xi-\xi')\rightarrow 0.$ 

Since  the scalar product \eqref{a10} vanishes in the limit $(\xi-\xi')\rightarrow \infty$, the states become orthogonal, the quantization follows from the condition 
\begin{eqnarray}\label{reslt11}
	\pi\frac{\xi-\xi'}{\tau\hbar \sqrt{3}}&=&n\pi\cr \xi-\xi'&=&\xi_n=\tau \hbar \sqrt{3}n, \quad n\in \mathbb{Z}.
\end{eqnarray}
One notices that the spectrum of momentum operator $\hat p$ presents 
discrete values. From the latter equation, one sees that 
\begin{eqnarray} \label{reslt1}
	\delta \xi_n=\xi_{n+1}-\xi_n= \tau \hbar \sqrt{3}=\sqrt{3}\Delta P_{min}.
\end{eqnarray}
With the  above results \eqref{reslt11} and \eqref{reslt1} at hand, one confirms that the formal momentum eigenvectors  $ \big|\phi_{\xi_n}\big\rangle =\big|\phi_{\tau \hbar \sqrt{3}n}\big\rangle $  are physically accepted and relevant. One  may be  tempted to interpret  the result (\ref{reslt1}) as if we are  describing physics on a
lattice in which   each  sites are spacing by the value $\sqrt{3}\Delta P_{min}$  illustrated as follows
	\begin{eqnarray*}
	\begin{tikzpicture}[node distance=3cm and 4cm]
		\node[] (O) at (-2,0) {$ $};
		\node[] (A) at (0,0) {$ \bullet$};
		\node (B) at (2,0) {$\bullet$};
		\node (C) at (4,0) {$\bullet$};	
		\node (D) at (6,0) {$\bullet$};
		\node (E) at (8,0) {$\bullet$};	
		\node (F) at (10,0) {$ $};
		\path[dashed,-] (O) edge (A);
		\path[-]  (A) node[below]{$\xi_1$}(A)  edge  node[above] {$\sqrt{3}\Delta P_{min}$} (B);
		\path[-] (B) node[below]{$\xi_2$} (B)  edge  node[above] {$\sqrt{3}\Delta P_{min}$}  (C);
			\path[-] (C) node[below]{$\xi_3$} (C) edge  node[above] {$\sqrt{3}\Delta P_{min}$}  (D);
			\path[-] (D) node[below]{$\xi_4$}(D) edge  node[above] {$\sqrt{3}\Delta p_{min}$}  (E);
		\path[dashed,-] (E)node[below]{$\xi_5$}(E) edge   (F);	
	\end{tikzpicture}	      
\end{eqnarray*}


\subsection{ Fourier transform and its inverse 	representations }
Since the states $|\phi_\xi\rangle$ are  physically meaningful and are well localized,  one can determine its Fourier transform (FT) and its inverse   representations  by projecting  an arbitrary state $|\psi\rangle$.

Let us propose a $\tau$-deformed coordinate $\Theta_\tau(x)$ and the $\tau$-deformed Lebesgue measure $d\mu_\tau(x)$ defined by
\begin{eqnarray}
    \Theta_\tau(x)= \frac{2}{\tau\sqrt{3}}\left[ \arctan(\frac{2\tau x-1}{\sqrt{3}})+\frac{\pi}{6}\right]\quad \mbox{and}\quad d\mu_\tau(x)=\frac{dx}{1-\tau x+\tau^2 x^2},
\end{eqnarray}
such that
\begin{eqnarray}
    d\mu_\tau(x)=d  \Theta_\tau(x).
\end{eqnarray}

\begin{definition} {\it Let $ \mathcal{S}\left(\mathbb{R}\right)$ be the Schwarz space which is dense in $\mathcal{H}=\mathcal{L}^2\left(\mathbb{R}\right)$. Let $|\psi\rangle \in\mathcal{S}\left(\mathbb{R}\right)$, the FT denoted by $ \mathcal{F}_\tau[\psi]$ or $\psi (\xi)$  is given by
	\begin{eqnarray}
		\psi (\xi)=\mathcal{F}_\tau[\psi(x)](\xi)
		=\sqrt{\frac{\tau\sqrt{3}}{\pi}}  \int_{-\ell_{max}}^{+\ell_{max}} d\mu_\tau(x)\psi(x) 
		e^{-\frac{i}{\hbar}\xi\Theta_\tau(x)}.\label{moment1}
	\end{eqnarray}
	The inverse  FT is given by 
	\begin{eqnarray}
		\psi(x)=\mathcal{F}_\tau^{-1}[\psi(\xi)](x)=\frac{1}{\hbar\sqrt{4\pi\tau\sqrt{3}}} \int_{-\infty}^{+\infty} d\xi \psi(\xi) e^{\frac{i}{\hbar}\xi\Theta_\tau(x)}.\label{F1}
\end{eqnarray}}
\end{definition}
\begin{remark} One can  notice that, for any $ \psi\in\mathcal{S}\left(\mathbb{R}\right)$, we have 
\begin{eqnarray}
\mathcal{F}_\tau^{-1}\mathcal{F}_\tau\psi=\psi=\mathcal{F}_\tau\mathcal{F}_\tau^{-1}\psi.
\end{eqnarray}
\end{remark}


\begin{proposition}\label{p21}
The following properties are derived from the definition of FT  
\begin{enumerate}
    \item \textbf{Linearity:} Let   $\psi, \chi\in \mathcal{S}\left(\mathbb{R}\right),$ we have
    \begin{eqnarray}
      \mathcal{F}_\tau\left[\alpha\psi(x)+\beta \chi(x)\right](\xi)=  \alpha \psi(\xi)+  \beta \chi(\xi),\quad \alpha,\beta\in\mathbb{C}.  
    \end{eqnarray}
\item \textbf{Plancherel theorem:} Let   $\psi \in \mathcal{S}\left(\mathbb{R}\right),$ we have
\begin{eqnarray}
    \frac{1}{2\hbar \tau \sqrt{3}} \int_{-\infty}^{+\infty}	 |\mathcal{F}_\tau[\psi(x)](\xi)|^2 d\xi=  \int_{-\ell_{max}}^{+\ell_{max}}	 |\psi(x)|^2 d\mu_\tau(x)
\end{eqnarray}
\item \textbf{Even (respectively odd) function property:} If $\psi(x)$ is an even (respectively odd) function,  then  $\mathcal{F}_\tau[\psi(x)](\xi)$ is even (respectively odd). More
generally, if we note $\mathcal{P}$ the operator that associates $ \mathcal{P}\psi(x)=\psi(-x)$, then  we have 
\begin{eqnarray}
\mathcal{F}_\tau[  \mathcal{P}\psi(x)](\xi)]= \mathcal{P}\psi(\xi),\quad \xi\in \mathbb{R}.
\end{eqnarray}
\item \textbf{Complex conjugation property:} If $\psi(x)$ is a real valued function, then we have 
\begin{eqnarray}
    \mathcal{F}_\tau[\psi(x)](-\xi)=\overline{\psi(\xi)},\quad \xi\in \mathbb{R}.
\end{eqnarray}
\item \textbf{Derivation and multiplication properties:}
\begin{eqnarray}
 \mathcal{F}_\tau[D_x\psi(x)](\xi)  &=&-\frac{i}{\hbar}\xi \psi(\xi),\quad \xi\in \mathbb{R}.\label{R1}\\ 
    \mathcal{F}_\tau[\Theta_\tau(x)\psi(x)](\xi)&=& i\hbar  \partial_\xi\psi(\xi), \quad \xi\in \mathbb{R},\label{Y1}
\end{eqnarray}
\item \textbf{Deformed translation:} Let $\theta_0\in \mathbb{R}$  and
$\psi_{\theta_0}(x)= \psi\left(\Theta_\tau(x)-\theta_0\right)$. Then
  \begin{eqnarray}
      \mathcal{F}_\tau[\psi_{\theta_0}(x)](\xi)=e^{-\frac{i}{\hbar}\theta_0\xi} \psi(\xi), \quad \xi\in \mathbb{R}.
  \end{eqnarray}
  and for $\xi_0\in \mathbb{R}$, such that $\psi_{\xi_0}(x)=e^{-\frac{i}{\hbar}\xi_0\Theta_\tau(x)} \psi(x)$. We  have 
   \begin{eqnarray}
      \mathcal{F}_\tau[\psi_{\xi_0}(x)](\xi)=\psi(\xi-\xi_0) , \quad \xi\in \mathbb{R}.
  \end{eqnarray}
  \item \textbf{Dilatation (Scaling) property:} Let $\alpha>0$ and $\psi_\alpha(x)=\psi(\alpha x )$. Then, we have
  \begin{eqnarray}
       \mathcal{F}_\tau[\psi_\alpha(x)](\xi)=\frac{1}{\sqrt{a}}\mathcal{F}_{\tau/a}[\psi(x)]\left(\frac{\xi}{a}\right), \quad \xi\in \mathbb{R}.
  \end{eqnarray}
\item \textbf{Convolution theorem:} Let $\psi,\chi\in  \mathcal{S}\left(\mathbb{R}\right)$. Let  $q=\Theta_\tau(x)$ such that $dq=d\Theta_\tau(x)=d\mu_\tau(x)$, $x=\Theta^{-1}(q)$ and 
  \begin{eqnarray*}
     \tilde{\psi}(q)=\psi(\Theta^{-1}(q))=\psi(x), \quad \tilde{\chi}(q)=\chi(\Theta^{-1}(q))=\chi(x).
  \end{eqnarray*}
  Their ordinary convolution in the $q$-coordinate is given by
  \begin{eqnarray}
      (\tilde{\psi}*\tilde{\chi})(q)=\int_{-\frac{\pi}{3\tau\sqrt{3}}}^{+\frac{2\pi}{3\tau\sqrt{3}}}dq’\tilde{\psi}(q’)\tilde{\chi})(q-q’).
  \end{eqnarray}
Then, the FT of this convolution is given by 
\begin{eqnarray}\label{ct}
    \mathcal{F}_\tau[\left(\psi*_\tau \chi\right)(x)](\xi)=\sqrt{\frac{\pi}{\tau\sqrt{3}}}\psi(\xi)\chi(\xi).
\end{eqnarray}
\end{enumerate}

\end{proposition}
\begin{proof}
See Appendix \eqref{ap2}	
\end{proof}
\begin{remark}
 Derivation and multiplication properties are mistakenly well established in our recent study \cite{37}. We correct this issue in this work by using the proposition \eqref{p21} with equations \eqref{R1} and \eqref{Y1}.
\end{remark}
From the FT of the convolution theorem \eqref{ct}, other useful identities are given in the following lemma
\begin{lemma}\label{lema2}
The FT of convolution has the usual algebraic properties:
\begin{enumerate}
    \item \textbf{Commutativity}: For $\psi, \chi\in \mathcal{S}\left(\mathbb{R}\right)$, we have
    \begin{eqnarray}
    \mathcal{F}_\tau[\left(\psi*_\tau \chi\right)(x)](\xi)= \mathcal{F}_\tau[\left(\chi*_\tau \psi\right)(x)](\xi).
\end{eqnarray}
\item \textbf{Associativity:} For $\psi, \chi,\varrho\in \mathcal{S}\left(\mathbb{R}\right)$
\begin{eqnarray}
    \mathcal{F}_\tau[\left(\left(\psi*_\tau \chi\right)*_\tau \varrho\right)(x)](\xi)= \mathcal{F}_\tau\left[\left(\psi*_\tau \left( \chi *_\tau \varrho\right)\right)(x)\right](\xi),
\end{eqnarray}
and 
\begin{eqnarray}
    \mathcal{F}_\tau[\left(\psi*_\tau \chi *_\tau \varrho\right)(x)](\xi)= \frac{\pi}{\tau\sqrt{3}}\psi(\xi)\chi(\xi)\varrho(\xi).
\end{eqnarray}
\item \textbf{Linearity:} For $\psi, \chi,\varrho\in \mathcal{S}\left(\mathbb{R}\right)$, we have
\begin{eqnarray}
    \mathcal{F}_\tau[\left(\left(\alpha\psi+ \beta\chi\right)*_\tau \varrho\right)(x)](\xi)= \alpha\mathcal{F}_\tau\left[\left(\left(\psi*_\tau \varrho\right)\right)(x)\right](\xi)+ \beta\mathcal{F}_\tau\left[\left(\left(\chi*_\tau \varrho\right)\right)(x)\right](\xi).
\end{eqnarray}
\item \textbf{Product theorem:} For $\psi, \chi\in \mathcal{S}\left(\mathbb{R}\right)$, we have
\begin{eqnarray}
    \mathcal{F}_\tau[\left(\psi\chi\right)(x)](\xi)= \frac{1}{2\hbar\sqrt{\pi\tau\sqrt{3}}}\int_{\mathbb{R}}d\eta \chi(\eta)\psi(\xi-\eta).
\end{eqnarray}
\item \textbf{Inversion:} For $\psi, \chi\in \mathcal{S}\left(\mathbb{R}\right)$, we have
\begin{eqnarray}
    \mathcal{F}_\tau^{-1}\left[ \mathcal{F}_\tau\psi  \mathcal{F}_\tau\chi\right](\xi)
    = \sqrt{\frac{\tau\sqrt{3}}{\pi}}\left(\psi*_\tau \chi\right)(x).
\end{eqnarray}
\end{enumerate}
\end{lemma}
\begin{proof}
See Appendix \eqref{PRO}	
\end{proof}

So far we have discussed the properties of the Fourier transform in  $ \mathcal{S}\left(\mathbb{R}\right)$. Since $ \mathcal{S}\left(\mathbb{R}\right)$ is included in $\mathcal{H}=\mathcal{L}^2\left(\mathbb{R}\right)$, it can be endowed with the
corresponding norm. The Fourier transform of  $ \mathcal{S}\left(\mathbb{R}\right)$ is then a quasi-isometry for this
norm:
\begin{proposition}\label{tz} Let $\psi,\chi\in \mathcal{S}\left(\mathbb{R}\right),$ then we have 
\begin{eqnarray}
    \frac{1}{2\hbar\tau\sqrt{3}}\int_{\mathbb{R}}\psi^*(\xi)\chi(\xi)d\xi= \int_{-\ell_{max}}^{+\ell_{max}} d\mu_\tau(x)\psi^*(x)\chi(x).
    \end{eqnarray}
In particular,
\begin{eqnarray}
    || \mathcal{F}_\tau\psi||^2=  2\hbar\tau\sqrt{3}||\psi||_{S_+}^2.
\end{eqnarray}
This proposition is identified as the well known the Parseval identity.
\end{proposition}
\begin{proof}
See Appendix\eqref{LMA}	
\end{proof}

\begin{lemma} \label{LMA2}
{\it The action of pseudo-Hermitian operators  on $\psi(\xi)$  reads as follows
	\begin{eqnarray}
		\hat X\psi(\xi)&=&\frac{2}{\tau}\frac{\tan\left(i\frac{\tau\hbar \sqrt{3}}{2}\partial_\xi\right)}{\sqrt{3}+\tan\left(i\frac{\tau\hbar \sqrt{3}}{2}\partial_\xi\right)}\psi(\xi),\\
		\hat P \psi(\xi)&=& \xi\psi(\xi).\label{qooo}
\end{eqnarray}}
\end{lemma}
\begin{proof}
See Appendix \eqref{ba22}	
\end{proof}
\begin{remark}
From the limit $\tau\rightarrow 0$ in the last equations, we recover the  ordinary  representations in momentum space as
\begin{eqnarray}
	\lim_{\tau\rightarrow 0}\hat X\psi(\xi)&=&i\hbar\partial_\xi\psi(\xi),\\
	\lim_{\tau\rightarrow 0}	\hat P \psi(\xi)&=&\xi\psi(\xi).
\end{eqnarray}
\end{remark}

Hereafter, with these Hilbert space representations at hand,  we study in what follows  the spectrum of a free particle and a harmonic oscillator  in position space representations.
\section{Spectrum of some solvable systems}
The general Hamiltonian of a free particle and a harmonic oscillator take the form 
\begin{eqnarray}\label{xc}
    \hat H= \frac{\hat P^2}{2m}+V(\hat X),
\end{eqnarray}
where $V(\hat X)$ is the potential energy. This Hamiltonian can be rewritten as
\begin{eqnarray} \label{Hami}
	\hat H=\frac{1}{2m} \left(\mathbb{I}-\tau \hat x_0+\tau^2\hat x_0^2\right) \hat p_{x_0}\left(\mathbb{I}-\tau \hat x_0+\tau^2\hat x_0^2\right) \hat p_{x_0} +V(\hat x_0),
\end{eqnarray}
where operators $\hat x_0,\hat y_0$ and $\hat p_{x_0},\hat p_{y_0}$ satisfy the commutation relations \eqref{alg1}.
\subsection{Spectrum of a free particle }\label{sec5}
 In this subsection, to understand the effect of this deformation on this system, we consider the system within  the following 2D Heisenberg algebra with the  deformation in $y$-direction   
  \begin{eqnarray}\label{alg4}
	[\hat X,\hat P_X]=i\hbar,\quad 	[\hat Y,\hat P_Y]=i\hbar \left(\mathbb{I}-\tau \hat Y+\tau^2\hat Y^2\right), \quad 	{[\hat X,\hat Y ]}=0=	{[\hat P_X,\hat P_Y ]},
 \end{eqnarray} 
 where the operators  $\hat X,\hat Y$ and $\hat P_X,\hat P_Y $ are  defined in $\mathcal{H}_\tau$ by
 \begin{eqnarray}\label{a24}
 	\hat X=\hat x_0,\, \,\hat Y=\hat y_0 \quad \mbox{and}\quad \hat P_X= \hat p_{x_0},\quad \hat P_Y= \left(\mathbb{I}-\tau \hat y_0+\tau^2\hat y_0^2\right)\hat p_{y_0}
 \end{eqnarray}

 The Hamiltonian of a free particle captured in a two-dimensional box of length $0<x_0<a$ and heigth $0<y_0<a$ involving the latter operators is given by      
\begin{eqnarray} \label{Hami}
	\hat H_F&=&\frac{1}{2m}\left(\hat p_{x_0}^2+ \left(\mathbb{I}-\tau \hat y_0+\tau^2\hat y_0^2\right) \hat p_{y_0}\left(\mathbb{I}-\tau \hat y_0+\tau^2\hat y_0^2\right) \hat p_{y_0}\right)\cr
	&=&\frac{1}{2m}\hat p_{x_0}^2 + \frac{1}{2m}\left[\left(\mathbb{I}-\tau \hat y_0+\tau^2\hat y_0^2\right) \hat p_{y_0} \right]^2, 
\end{eqnarray}
where the Hamiltonian of the system in $x$-direction is free from any deformation.
The time-independent Schr\"{o}dinger equation is given by
\begin{eqnarray}
E\phi(x,y)=-\frac{\hbar^2}{2m}\partial_{x_0}^2\phi(x,y)-\frac{\hbar^2}{2m}\left[\left(1-\tau  y_0+\tau^2 y_0^2\right)\partial_{y_0} \right]^2\phi(x,y)
\end{eqnarray}
As it is clearly seen, the system is decoupled and the solution to the eigenvalue equation 
is given by
\begin{eqnarray}
	\phi(x,y)= \phi(x)\phi(y), \quad E=E^x+E^y
\end{eqnarray}
where $\phi(x) $ is the wave function in the $x$-direction and $\phi(y) $  the wave function in the
$y$-direction. Since the particle is free in the $x$-direction and bounded in $0<x_0<a$, the wave function and its eigenvalues are given by \cite{36}
\begin{eqnarray}\label{eq8}
	\phi_n(x)=\sqrt{\frac{2}{a}}\sin\left(\frac{n\pi}{a}x\right), \quad E_{n}^x=n^2\frac{\pi^2\hbar^2}{2ma^2},\quad E_1^x=\frac{\pi^2\hbar^2}{2ma^2}.
\end{eqnarray}
Taking the results \eqref{eq8} as a witness, we study in what follows the influence of the
deformed parameter $\tau$ on the system.In $y$-direction, we have to solve  the following equation 
\begin{eqnarray}
-\frac{\hbar^2}{2m}\left[\left(1-\tau  y_0+\tau^2 y_0^2\right)\partial_{y_0} \right]^2\phi(y)=E^y\phi(y)\quad \mbox{with}\quad E^y>0
\end{eqnarray}
  In $y$-direction, the solution is given by
\begin{eqnarray}
	\phi_k (y)= \sqrt{\frac{\tau\sqrt{3}}{\pi}}\sin\left(\frac{2k}{\tau \sqrt{3}}\left[\arctan\left(\frac{2\tau y-1}{\sqrt{3}}\right)
	+\frac{\pi}{6}\right]\right),
\end{eqnarray}
where $k=\frac{\sqrt{2mE^y}}{\hbar}$. We suppose that, the wave function satisfies the Dirichlet condition i.e., it vanishes at the boundaries. Thus, 
the quantization follows from the boundary condition $\phi_k(a)=0$ and leads to the equation
\begin{eqnarray}
	\frac{2k_n}{\tau \sqrt{3}}\left[\arctan\left(\frac{2\tau a-1}{\sqrt{3}}\right)
	+\frac{\pi}{6}\right]=n\pi\quad \mbox{with}\quad n\in \mathbb{N}^*.
\end{eqnarray}
This equation becomes
\begin{eqnarray}\label{sp01}
	k_n= \frac{n\pi\tau \sqrt{3}}{2\left[\arctan\left(\frac{2\tau a-1}{\sqrt{3}}\right)
		+\frac{\pi}{6}\right]}\implies
	E_n^y	= \frac{3 n^2\pi^2\tau^2\hbar^2 }{8m \left[\arctan\left(\frac{2\tau a-1}{\sqrt{3}}\right)
		+\frac{\pi}{6}\right]^2}.\label{E}
\end{eqnarray}
At the limit $\tau\rightarrow 0$
, we have
\begin{eqnarray}\label{sp1}
	\lim_{\tau\rightarrow 0}E_{n}^y= E_{n}^x=n^2\frac{\pi^2\hbar^2}{2ma^2} 
\end{eqnarray}
Thus, the energy levels can be rewritten as
\begin{eqnarray}\label{sp01}
	E_n^y	= \frac{3}{4}\left[\frac{\tau a}{\arctan\left(\frac{2\tau a-1}{\sqrt{3}}\right)
		+\frac{\pi}{6}}\right]^2E_n^x\leq E_n^x
\end{eqnarray}
 This result shows that the deformation parameter $\tau$ induces weak  energy transitions in y-direction.  This
indicates that, the deformed parameter induces a more pronounced contraction of energy levels which, consequently implies the decrease of energy band structures.

Finally, the total energy is given by
\begin{eqnarray}\label{sp1}
	E=E_{n}^x+ E_{n}^y=\left(1+\frac{3\tau^2 a}{\arctan\left(\frac{2\tau a-1}{\sqrt{3}}\right)
		+\frac{\pi}{6}}\right) E_{n}^x
\end{eqnarray}
At the limit $\tau\rightarrow 0,$ we have 
\begin{eqnarray}\label{sp1}
	\lim_{\tau\rightarrow 0}E=2 E_{n}^x=n^2\frac{\pi^2\hbar^2}{ma^2}.
\end{eqnarray}

\subsection{Spectrum of a harmonic oscillator}
In this section,  we consider 1D deformed Heisenberg algebra  given by
\begin{eqnarray}\label{alg4}
 	[\hat X,\hat P_X]=i\hbar f\left(\hat X\right).
 \end{eqnarray} 
  where  operators  $\hat X, \hat P_X$ and $f\left(\hat X\right)$  are  defined by
 \begin{eqnarray}\label{a24}
 	\hat X=\hat x, \quad f\left(\hat X\right)=1-\tau \hat x +\tau^2 \hat x^2,\quad \mbox{and}\quad \hat P_X=f(\hat x)\hat p, \quad \mbox{with}\quad [\hat x,\hat p]=i\hbar.
 \end{eqnarray}
 The Schrödinger equation  of the   Hamiltonian \eqref{xc} within this deformed-algebra is defined by
 \begin{eqnarray}\label{wt}
     \hat H\psi(x)=-\frac{\hbar^2}{2m}\left(f( x)\frac{d}{dx}\right)^2\psi(x)+\frac{m\omega^2 }{2}x^2\psi(x)=E\psi(x).
 \end{eqnarray}
Let  define the change of variable by
\begin{equation}
f(x)\frac{d}{dx}=\frac{d}{d\rho}.
\end{equation}
The latter Schrödinger \eqref{wt} equation is reduced into \cite{38}
\begin{equation}
\left[
-\frac{d^2}{d\rho^2}
+\frac{m^2\omega^2}{\hbar^2}x^2(\rho)
\right]\psi(\rho)
=
\frac{2mE}{\hbar^2}\psi(\rho).
\end{equation}
From the transformation
\begin{equation}
\frac{d}{d\rho}
=
\left(1-\tau x+\tau^2x^2\right)\frac{d}{dx}\implies \frac{d\rho}{dx}
=\frac{1}{1-\tau x+\tau^2x^2},
\end{equation}
we have
\begin{align}
\rho=
\frac{2}{\sqrt{3}\tau}
\tan^{-1}\left(\frac{2\tau x-1}{\sqrt{3}}
\right)\quad \mbox{and}\quad x(\rho)
=
\frac{1+\sqrt{3}
\tan\left(\frac{\sqrt{3}\tau \rho}{2}\right)}
{2\tau}
\end{align}
The Schrödinger equation becomes 
\begin{equation}\label{dr}
\left[-\frac{d^2}{d\theta^2}+K\left(1+\sqrt{3}\tan\theta\right)^2
\right]\psi(\theta)=\frac{8mE}{3\hbar^2\tau^2}\psi(\theta).
\end{equation}
where 
\begin{eqnarray}    
\theta=\frac{\sqrt{3}\tau }{2}\rho\quad \mbox{and}\quad\quad K=
\frac{m^2\omega^2}
{3\hbar^2\tau^4}.
\end{eqnarray}
Expanding equation \eqref{dr}, we have
\begin{equation} \label{rw}
\left[
-\frac{d^2}{d\theta^2}
+3K\csc^2\theta
+2\sqrt{3}K\tan\theta\right]\psi(\theta)
=
\left(\frac{8mE}{3\hbar^2\tau^2}-K\right)\psi(\theta).
\end{equation}
Referring to the standard trigonometric Rosen-Morse equation \cite{40} such that 
\begin{equation}\label{rw1}
\left[
-\frac{d^2}{d\theta^2}+A(A-1)\csc^2\theta-2B\cot\theta\right]\psi(\theta)
=
\epsilon\psi(\theta).
\end{equation}
By comprising equations \eqref{rw} and \eqref{rw1} give
\begin{equation}
A(A-1)=3K, \quad B=\sqrt{3}K \quad \mbox{and}\quad \epsilon=\frac{8mE}{3\hbar^2\tau^2}
-K
\end{equation}
Therefore,
\begin{equation}
A=\frac{1}{2}\left(1+
\sqrt{1+\frac{4m^2\omega^2}
{\hbar^2\tau^4}}\right),\quad  K=\frac{m^2\omega^2}{3\hbar^2\tau^4},
\end{equation}
For the trigonometric Rosen-Morse potential \eqref{rw1}, the eigenvalues are
\begin{equation}
\epsilon_n=(n+A)^2-\frac{B^2}{(n+A)^2},
\qquad
n=0,1,2,\ldots
\end{equation}
Thus, we have 
\begin{equation}
\frac{8mE_n}{3\hbar^2\tau^2}-K=(n+A)^2-\frac{B^2}{(n+A)^2}.
\end{equation}
Since $B^2=3K^2$, the energy eigenvalues are
\begin{eqnarray}
E_n=\frac{3\hbar^2\tau^2}{8m}\left[
(n+A)^2-\frac{3K^2}{(n+A)^2}-K\right],\qquad=0,1,2,\ldots.
\end{eqnarray}
The ground state energy $E_0$ is obtained by
\begin{eqnarray}
E_0=\frac{3\hbar^2\tau^2}{8m}\left(
A^2-\frac{3K^2}{A^2}-K\right).
\end{eqnarray}
 At the limit $\tau\rightarrow 0$, we recover the eigenvalues of ordinary harmonic oscillator
 \begin{eqnarray}
     \lim_{\tau\rightarrow 0} E_n=\hbar \omega\left(n+\frac{1}{2}\right).
 \end{eqnarray}
The eigenfunctions of the trigonometric Rosen--Morse potential can be written in terms of Jacobi polynomials \cite{39,40} as
\begin{equation}\label{trq}
\psi_n(\theta)
=N_n(\sin\theta)^{n+A}\exp\left[-\frac{B}{n+A}\theta\right]P_n^{(\alpha_n,\beta_n)}(i\cot\theta),
\end{equation}
where $P_n^{(\alpha_n,\beta_n)}$ is the Jacobi polynomial and  $N_n$ is the normalization constant and
\begin{equation}
\alpha_n
=-(n+A)+\frac{iB}{n+A},\quad \beta_n=-(n+A)-\frac{iB}{n+A}.
\end{equation}
Since $B=\sqrt{3}K$, this becomes
\begin{equation}
\begin{aligned}
\psi_n(\theta)
={}&N_n
(\sin\theta)^{n+A}
\exp\left[
-\frac{\sqrt{3}K}{n+A}\theta
\right]
\\
&\times
P_n^{
\left(
-(n+A)+\frac{i\sqrt{3}K}{n+A},
-(n+A)-\frac{i\sqrt{3}K}{n+A}
\right)}
(i\cot\theta).
\end{aligned}
\end{equation}
Finally, using
$\theta
=\frac{\sqrt{3}\tau \rho}{2}$, the eigenfunctions in the $\rho$ representation are obtained as
\begin{equation}
\psi_n(\rho)=\psi_n\left(
\frac{\sqrt{3}\tau \rho}{2}\right).
\end{equation}
For $n=0$, we have
\begin{equation}
P_0^{(\alpha_0,\beta_0)}(z)=1.
\end{equation}
Therefore, the ground states are given by
\begin{eqnarray}
\psi_0(\theta)
&=&
N_0
(\sin\theta)^A
\exp\left(
-\frac{\sqrt{3}K}{A}\theta
\right),
\\
\psi_0(\rho)&=&N_0\left[\sin\left(\frac{\sqrt{3}\tau p}{2}\right)\right]^A\exp\left[-\frac{\sqrt{3}K}{A}
\left(\frac{\sqrt{3}\tau \rho}{2}\right)\right].
\end{eqnarray}

From the  eigenfunction\eqref{trq}, one can  expresses it in terms of Laguerre polynomial. The Jacobi-to-Laguerre is  obtained from \cite{39,41}
\begin{eqnarray}
    \lim_{\beta\rightarrow \infty} P_n^{(\alpha,\beta)}\left(1-\frac{2z}{\beta}\right)= L^{(\alpha)}(z),
\end{eqnarray}
where $L^{(\alpha)}(z)$ is Laguerre  polynomial. In the present context, if  we introduce scaling  such that $i \cot\theta=1-\frac{2z}{\beta_n}$, consequently the Rosen–Morse eigenfunction has the limiting structure
\begin{eqnarray}
\psi_n(\theta)=N_n(\sin\theta)^{n+A}\exp\left(-\frac{\sqrt{3}K}{n+A}\theta\right) L_n^{\left(-(n+A)+\frac{i\sqrt{3}K}{n+A}\right)}(z).
\end{eqnarray}


 \section{Conclusion}\label{sec6}
For dynamical quantum systems to be studied, the Hamiltonian operator must be Hermitian. Thus, the Hermiticity of the Hamiltonian ensures the realism of the spectrum, the conservation of probability density, and the orthoganility of the Hamiltonian eigenbasis.  We have shown in the present work that a Hamiltonian operator with real spectrum is no longer Hermitian within a position-deformed Heisenberg algebra \eqref{alg41}. Using a pseudo-similarity transformation and a suitable positive-definite pseudo-metric operator, we have established the Hermiticity of this operator. Next, we constructed Hilbert space representations  associated with the pseudo-Hermitian operators. We  finally studied the eigenvalue problem of a free particle  and an ordinary harmonic oscillator in this deformed space in this deformed  space. We  have particularly show in the case of a free particle system that, this deformation strongly curved the  quantum levels allowing particles to jump from one state to another with low energy transitions.

\section*{Declarations}
\begin{itemize}
    \item \textbf{Funding} : This research did not receive funding.
    \item \textbf{Conflict of interest:}  The authors declare that there is no conflict of interest regarding
the publication of this paper.
\item  \textbf{ Data availability}: No datasets were generated or analyzed during the current study.
    
\end{itemize}


 

 \section*{Appendix: Proof of Propositions and Lemmas}

 \subsection*{Proof of the proposition\eqref{p21}}\label{ap2}
 Let us consider a $\tau$-deformed coordinate $\Theta_\tau(x)$ and the $\tau$-deformed Lebesgue measure $d\mu_\tau(x)$ defined by
\begin{eqnarray*}
    \Theta_\tau(x)= \frac{2}{\tau\sqrt{3}}\left[ \arctan(\frac{2\tau x-1}{\sqrt{3}})+\frac{\pi}{6}\right]\quad \mbox{and}\quad d\mu_\tau(x)=\frac{dx}{1-\tau x+\tau^2 x^2},
\end{eqnarray*}
such that
\begin{eqnarray*}
    d\mu_\tau(x)=d  \Theta_\tau(x).
\end{eqnarray*}

\begin{enumerate}
    \item 
    \begin{proof} For $\alpha,\beta\in \mathbb{C}$, we have 
 	\begin{eqnarray*}
 		\mathcal{F}_\tau[\alpha\psi(x)+\beta \chi(x)](\xi)&=& \sqrt{\frac{\tau\sqrt{3}}{\pi}}  \int_{-\ell_{max}}^{+\ell_{max}}\alpha\psi(x) 
 		e^{-\frac{i}{\hbar}\xi  \Theta_\tau(x)}d\mu_\tau(x)\cr&&+ \sqrt{\frac{\tau\sqrt{3}}{\pi}}  \int_{-\ell_{max}}^{+\ell_{max}}\beta\chi(x) 
 		e^{-\frac{i}{\hbar}\xi  \Theta_\tau(x)}d\mu_\tau(x)\cr
 		&=&\alpha\sqrt{\frac{\tau\sqrt{3}}{\pi}}  \int_{-\ell_{max}}^{+\ell_{max}}\psi(x) 
 		e^{-\frac{i}{\hbar}\xi  \Theta_\tau(x)}d\mu_\tau(x)\cr&&+ \beta\sqrt{\frac{\tau\sqrt{3}}{\pi}}  \int_{-\ell_{max}}^{+\ell_{max}}\chi(x) 
 		e^{-\frac{i}{\hbar}\xi  \Theta_\tau(x)}d\mu_\tau(x)\\
 		&=&\alpha \psi(\xi)+\beta \chi(\xi).
 	\end{eqnarray*}
   \end{proof}

    \begin{proof}
 	\item From the FT, we have 
 	\begin{eqnarray*}
 		\int_{-\infty}^{+\infty}	 |\mathcal{F}_\tau[\psi(x)](\xi)|^2 d\xi&=& \int_{-\infty}^{+\infty}|\psi(\xi)|^2d\xi=\int_{-\infty}^{+\infty}\psi(\xi)\psi^*(\xi)d\xi\cr
 		&=&\sqrt{\frac{\tau\sqrt{3}}{\pi}}\int_{-\infty}^{+\infty}d\xi\psi(\xi)\cr&&\times\left[  \int_{-\ell_{max}}^{+\ell_{max}}d\mu_\tau(x) \psi^*(x)
 		e^{i\xi\Theta_\tau(x) }  \right]\cr
 		&=&\sqrt{\frac{\tau\sqrt{3}}{\pi}}\int_{-\ell_{max}}^{+\ell_{max}}d\mu_\tau(x)\psi^*(x)\cr&&\times\int_{-\infty}^{+\infty} \psi(\xi)d\xi
 		e^{i\frac{2\xi}{\tau \hbar \sqrt{3}}\left[\arctan\left(\frac{2\tau x-1}{\sqrt{3}}\right)
 			+\frac{\pi}{6}\right]}\cr
 		&=&2\hbar \tau \sqrt{3}\int_{-\ell_{max}}^{+\ell_{max}}\psi^*(x)\psi(x)d\mu_\tau(x)\cr
 		&=&2\hbar \tau \sqrt{3}\int_{-\ell_{max}}^{+\ell_{max}} |\psi(x)|^2 d\mu_\tau(x)
 	\end{eqnarray*}
\end{proof}

\item
 \begin{proof} Let $\psi(x)$ be an even (respectively odd) function and  $\mathcal{P}$ be an operator that associates $ \mathcal{P}\psi(x)=\psi(-x)$, then we have 
\begin{eqnarray*}
\mathcal{F}_\tau[  \mathcal{P}\psi(x)](\xi)]= \sqrt{\frac{\tau\sqrt{3}}{\pi}}  \int_{-\ell_{max}}^{+\ell_{max}}\psi(-x) 
 		e^{-\frac{i}{\hbar}\xi  \Theta_\tau(x)}d\mu_\tau(x).
\end{eqnarray*}
For $\psi(x)$  is an even (respectively odd) function, we have $  \mathcal{P}\psi(x)=\psi(x)$ (respectively $\mathcal{P}\psi(x)=-\psi(x)$). Consequently, $\psi(\xi)=\mathcal{F}_\tau[\psi(x)](\xi)]$ is even (respectively odd) such that 
\begin{eqnarray*}
\mathcal{F}_\tau[  \mathcal{P}\psi(x)](\xi)]= \mathcal{P}\psi(\xi).
\end{eqnarray*}
\end{proof}
\item  \begin{proof}
Let  $\psi(x)$  is a real valued function, then we have
\begin{eqnarray*}
    \mathcal{F}_\tau[ \psi(x)](-\xi)]&=&\sqrt{\frac{\tau\sqrt{3}}{\pi}}  \int_{-\ell_{max}}^{+\ell_{max}}\psi(x) 
 		e^{-\frac{i}{\hbar}(-\xi)  \Theta_\tau(x)}d\mu_\tau(x)\\
       &=&  \sqrt{\frac{\tau\sqrt{3}}{\pi}}  \int_{-\ell_{max}}^{+\ell_{max}}\psi(x) 
 		e^{\frac{i}{\hbar}\xi  \Theta_\tau(x)}d\mu_\tau(x)\\
       &=& \overline{\psi(\xi)} 
\end{eqnarray*}
\end{proof}

\item  \begin{proof}
Let consider  $D_x=(1-\tau x+\tau^2 x^2)\frac{d}{dx}$ and $u=\Theta_\tau(x)$, we have
\begin{eqnarray*}
 \frac{d}{du}=\frac{dx}{du}  \frac{d}{dx} =D_x=\frac{d}{d\Theta_\tau(x)}.
\end{eqnarray*}
Then,
\begin{eqnarray*}
	\mathcal{F}_\tau[D_x\psi(x)](\xi)
		&=&\sqrt{\frac{\tau\sqrt{3}}{\pi}}  \int_{-\ell_{max}}^{+\ell_{max}}d\mu_\tau(x)\psi(x) \frac{de^{-\frac{i}{\hbar}\xi \Theta_\tau(x)}}{d\Theta_\tau(x)}\\
           &=&   -\frac{i}{\hbar}\xi\mathcal{F}_\tau[\psi(x)](\xi)\\
	\end{eqnarray*}
\begin{eqnarray*}
    	\mathcal{F}[\Theta_\tau(x)\psi(x)](\xi) &=& \sqrt{\frac{\tau\sqrt{3}}{\pi}}  \int_{-\ell_{max}}^{+\ell_{max}}\psi(x) 
 		\Theta_\tau(x)e^{-\frac{i}{\hbar}\xi  \Theta_\tau(x)}d\mu_\tau(x)\\
        &=&i\hbar \frac{d}{d\xi} \sqrt{\frac{\tau\sqrt{3}}{\pi}}  \int_{-\ell_{max}}^{+\ell_{max}}\psi(x) 
 		e^{-\frac{i}{\hbar}\xi  \Theta_\tau(x)}d\mu_\tau(x)\\
          &=&i\hbar \frac{d}{d\xi} \mathcal{F}[\psi(x)](\xi)
    \end{eqnarray*}
    \end{proof}
\item
 \begin{proof}
 Let $\theta_0\in \mathbb{R}$  and
$\psi_{\theta_0}(x)= \psi\left(\Theta_\tau(x)-\theta_0\right)$. Then
  \begin{eqnarray*}
      \mathcal{F}_\tau[\psi_{\theta_0}(x)](\xi)=\sqrt{\frac{\tau\sqrt{3}}{\pi}}  \int_{-\ell_{max}}^{+\ell_{max}}\psi\left(\Theta_\tau(x)-\theta_0\right) e^{-\frac{i}{\hbar}\xi  \Theta_\tau(x)}d\mu_\tau(x).
  \end{eqnarray*}
Let $u=\Theta_\tau(x)-\theta_0$ such that $du=d\Theta_\tau(x)=d\mu_\tau(x)$, then we have 
    \begin{eqnarray*}
      \mathcal{F}_\tau[\psi_a(x)](\xi)&=&\sqrt{\frac{\tau\sqrt{3}}{\pi}}  \int_{-\ell_{max}}^{+\ell_{max}}\psi\left(u\right) e^{-\frac{i}{\hbar}\xi  \theta_0}e^{-\frac{i}{\hbar}\xi  \Theta_\tau(x)}d\mu_\tau(x)\\
     &=& e^{-\frac{i}{\hbar}\xi  \theta_0}\sqrt{\frac{\tau\sqrt{3}}{\pi}}  \int_{-\ell_{max}}^{+\ell_{max}}\psi\left(u\right) e^{-\frac{i}{\hbar}\xi  \Theta_\tau(x)}d\mu_\tau(x)\\
     &=&e^{-\frac{i}{\hbar}\xi  \theta_0}\mathcal{F}_\tau[\psi(x)](\xi).
  \end{eqnarray*}

 For $\xi_0\in \mathbb{R}$, such that $\psi_{\xi_0}(x)=e^{-\frac{i}{\hbar}\xi_0\Theta_\tau(x)} \psi(x)$. We  have 
   \begin{eqnarray*}
      \mathcal{F}_\tau[\psi_{\xi_0}(x)](\xi)&=& \sqrt{\frac{\tau\sqrt{3}}{\pi}}  \int_{-\ell_{max}}^{+\ell_{max}}\psi\left(x\right)  e^{-\frac{i}{\hbar}\xi_0\Theta_\tau(x)}e^{-\frac{i}{\hbar}\xi  \Theta_\tau(x)}d\mu_\tau(x)\\
      &=& \sqrt{\frac{\tau\sqrt{3}}{\pi}}  \int_{-\ell_{max}}^{+\ell_{max}}\psi\left(x\right) e^{-\frac{i}{\hbar}(\xi -\xi_0) \Theta_\tau(x)}d\mu_\tau(x)\\
      &=& \mathcal{F}_\tau[\psi(x)](\xi-\xi_0)
  \end{eqnarray*}
 \end{proof}
\item \begin{proof}
Let $ a>0$ and $\psi_a(x)=\psi(a x)$. Then, we have
  \begin{eqnarray*}
       \mathcal{F}_\tau[\psi_a(x)](\xi)=\sqrt{\frac{\tau\sqrt{3}}{\pi}}  \int_{-\ell_{max}}^{+\ell_{max}}\psi\left(ax\right)   e^{-\frac{i}{\hbar}\xi  \Theta_\tau(x)}d\mu_\tau(x).
  \end{eqnarray*}
Let $u=a x$ such that $ du=a dx$ and 
\begin{eqnarray*}
    d\mu_{\tau_a}(u)=\frac{du}{1-\tau_a u+\tau_a^2 u^2}, \quad \Theta_{\tau}(u)=\frac{2}{\tau_a \sqrt{3}}\left[\arctan\left(\frac{2\tau_a u-1}{\sqrt{3}}\right)+\frac{\pi}{6}\right],\quad \mbox{with}\quad \tau_a=\frac{\tau}{a}.
\end{eqnarray*}
Then we have
\begin{eqnarray*}
       \mathcal{F}_\tau[\psi_a(x)](\xi)&=&\frac{1}{a^{1/2}}\sqrt{\frac{\tau_a\sqrt{3}}{\pi}}  \int_{-a\ell_{max}}^{+a\ell_{max}}\psi\left(u\right) e^{-\frac{i}{\hbar}\frac{\xi}{a}  \Theta_{\tau_a}(u)}d\mu_{\tau_a}(u)\\
       &=&\frac{1}{a^{1/2}}\mathcal{F}_{\tau_a} \left[\psi\right]\left(\frac{\xi}{a}\right),
  \end{eqnarray*}
 \end{proof} 
  \item 
   \begin{proof}
      Let  $q=\Theta_\tau(x)$ such that $dq=d\Theta_\tau(x)=d\mu_\tau(x)$ and $x=\Theta^{-1}(q)$. As $x\rightarrow -\ell_{max},\quad q \rightarrow -\frac{\pi}{3\tau\sqrt{3}}$ while $x\rightarrow \ell_{max},\quad q \rightarrow \frac{2\pi}{3\tau\sqrt{3}}$. Let
  \begin{eqnarray*}
     \tilde{\psi}(q)=\psi(\Theta^{-1}(q))=\psi(x), \quad \tilde{\chi}(q)=\chi(\Theta^{-1}(q))=\chi(x),
  \end{eqnarray*}
  such that their ordinary convolution in the $q$-coordinate is
  \begin{eqnarray*}
      (\tilde{\psi}*\tilde{\chi})(q)=\int_{-\frac{\pi}{3\tau\sqrt{3}}}^{+\frac{2\pi}{3\tau\sqrt{3}}}dq'\tilde{\psi}(q')\tilde{\chi})(q-q').
  \end{eqnarray*}
  Then
  \begin{eqnarray*}
     \mathcal{F}_\tau[ (\psi*_\tau\chi)(x)](\xi)&=&\sqrt{\frac{\tau\sqrt{3}}{\pi}} \int_{-\frac{\pi}{3\tau\sqrt{3}}}^{+\frac{2\pi}{3\tau\sqrt{3}}}dq'(\tilde{\psi}*\tilde{\chi})(q) e^{-\frac{i}{\hbar}q\xi}\cr
     &=& \sqrt{\frac{\tau\sqrt{3}}{\pi}} \int_{-\frac{\pi}{3\tau\sqrt{3}}}^{+\frac{2\pi}{3\tau\sqrt{3}}}dq \int_{-\frac{\pi}{3\tau\sqrt{3}}}^{+\frac{2\pi}{3\tau\sqrt{3}}}dq'\tilde{\psi}(q')\tilde{\chi}(q-q')e^{-\frac{i}{\hbar}q’\xi}\\
     &=& \sqrt{\frac{\tau\sqrt{3}}{\pi}}\int_{-\frac{\pi}{3\tau\sqrt{3}}}^{+\frac{2\pi}{3\tau\sqrt{3}}}dq'\tilde{\psi}(q')e^{-\frac{i}{\hbar}q’\xi} \int_{-\frac{\pi}{3\tau\sqrt{3}}}^{+\frac{2\pi}{3\tau\sqrt{3}}}du\tilde{\chi}(u)e^{-\frac{i}{\hbar}u\xi}.
  \end{eqnarray*}
  But 
  \begin{eqnarray*}
   \mathcal{F}_\tau[ (\psi)(x)](\xi)= \sqrt{\frac{\tau\sqrt{3}}{\pi}}\int_{-\frac{\pi}{3\tau\sqrt{3}}}^{+\frac{2\pi}{3\tau\sqrt{3}}}dq\tilde{\psi}(q)e^{-\frac{i}{\hbar}\xi q};\quad 
   \mathcal{F}_\tau[ (\chi)(x)](\xi)= \sqrt{\frac{\tau\sqrt{3}}{\pi}}\int_{-\frac{\pi}{3\tau\sqrt{3}}}^{+\frac{2\pi}{3\tau\sqrt{3}}}du\tilde{\psi}(u)e^{-\frac{i}{\hbar}\xi u}.
  \end{eqnarray*}

    Therefore 
    \begin{eqnarray*}
     \mathcal{F}_\tau[ (\psi*_\tau\chi)(x)](\xi)=\sqrt{\frac{\pi}{\sqrt{\tau\sqrt{3}}}} \mathcal{F}_\tau[ (\psi)(x)](\xi) \mathcal{F}_\tau[ (\chi)(x)](\xi)=\sqrt{\frac{\pi}{\sqrt{\tau\sqrt{3}}}}\psi(\xi)\chi(\xi).
  \end{eqnarray*}
 \end{proof} 
 
  \end{enumerate} 
\subsection*{Proof of the Lemma\eqref{lema2}}\label{LMA}
In this subsection, we give the proofs of some properties of Fourier transform of convolution theorem established in the lemma\eqref{lema2}
\begin{enumerate}
    \item 
    \begin{proof}
    Referring to the convolution theorem such that for  $\tilde{\psi}, \tilde{\chi}\in \mathcal{S}\left(\mathbb{R}\right)$, we have 
    \begin{eqnarray*}
      (\tilde{\psi}*\tilde{\chi})(q)=\int_{-\frac{\pi}{3\tau\sqrt{3}}}^{+\frac{2\pi}{3\tau\sqrt{3}}}dq’\tilde{\psi}(q’)\tilde{\chi}(q-q’).
  \end{eqnarray*}
  We have 
  \begin{eqnarray}
      (\tilde{\chi}*\tilde{\psi})(q)=\int_{-\frac{\pi}{3\tau\sqrt{3}}}^{+\frac{2\pi}{3\tau\sqrt{3}}}dq’\tilde{\chi}(q’)\tilde{\psi}(q-q’).
  \end{eqnarray}
   Let $q"=q-q’$ and $dq"=dq’$ we have 
  \begin{eqnarray*}
      (\tilde{\chi}*\tilde{\psi})(q)=\int_{-\frac{\pi}{3\tau\sqrt{3}}}^{+\frac{2\pi}{3\tau\sqrt{3}}}dq"\tilde{\psi}(q)\tilde{\chi}(q-q")=  (\tilde{\psi}*\tilde{\chi})(q).
  \end{eqnarray*}
 From this commutativity  of the convolution theorem, it is straightforward to show that the Fourier transform of the convolution is  also  commutative. Hence  for $\psi, \chi\in \mathcal{S}\left(\mathbb{R}\right)$, we have
 \begin{eqnarray*}
     \mathcal{F}_\tau[ (\psi*_\tau\chi)(x)](\xi)=  \mathcal{F}_\tau[ (\chi *_\tau \psi)(x)](\xi)
 \end{eqnarray*}
  \end{proof}
\item
\begin{proof}
Recall:
\begin{eqnarray*}
      (\tilde{\psi}*\tilde{\chi})(q)=\int_{-\frac{\pi}{3\tau\sqrt{3}}}^{+\frac{2\pi}{3\tau\sqrt{3}}}dq’\tilde{\psi}(q’)\tilde{\chi}(q-q’), \quad \tilde{\psi}(q),  \tilde{\chi}(q)\in \mathcal{S}\left(\mathbb{R}\right).
  \end{eqnarray*}
  Start with the left-hand side:
  \begin{eqnarray*}
      ((\tilde{\psi}*\tilde{\chi})* \tilde{\varrho})(q)&=&\int_{-\frac{\pi}{3\tau\sqrt{3}}}^{+\frac{2\pi}{3\tau\sqrt{3}}}dq’(\tilde{\psi}*\tilde{\chi})(q’)\tilde{\varrho}(q-q’)\\
      &=& \int_{-\frac{\pi}{3\tau\sqrt{3}}}^{+\frac{2\pi}{3\tau\sqrt{3}}}dq’ \left[\int_{-\frac{\pi}{3\tau\sqrt{3}}}^{+\frac{2\pi}{3\tau\sqrt{3}}}dq" \tilde{\psi}(q")\tilde{\chi}(q’-q")\right]\tilde{\varrho}(q-q’).
  \end{eqnarray*}
  By Fubini's theorem, we have
  \begin{eqnarray*}
      ((\tilde{\psi}*\tilde{\chi})* \tilde{\varrho})(q)=
    \int_{-\frac{\pi}{3\tau\sqrt{3}}}^{+\frac{2\pi}{3\tau\sqrt{3}}}dq’dq" \tilde{\psi}(q")\tilde{\chi}(q’-q")\tilde{\varrho}(q-q’).
  \end{eqnarray*}

 Now introduce the change of variables
 \begin{eqnarray}
     u=q",\quad v=q’-q" ,\quad q’=u+v,\quad \mbox{and} \quad dq"dq’=dudv.
 \end{eqnarray}
 Therefore
  \begin{eqnarray*}
      ((\tilde{\psi}*\tilde{\chi})* \tilde{\varrho})(q)=
    \int_{-\frac{\pi}{3\tau\sqrt{3}}}^{+\frac{2\pi}{3\tau\sqrt{3}}}dudv \tilde{\psi}(u)\tilde{\chi}(v)\tilde{\varrho}(q-u-v).
  \end{eqnarray*}
  Now calculate the right-hand side:
  \begin{eqnarray*}
      (\tilde{\psi}*(\tilde{\chi}* \tilde{\varrho}))(q)&=&\int_{-\frac{\pi}{3\tau\sqrt{3}}}^{+\frac{2\pi}{3\tau\sqrt{3}}}dq’\tilde{\psi}(q’)(\tilde{\chi}*\tilde{\varrho})(q-q')\\
      &=&\int_{-\frac{\pi}{3\tau\sqrt{3}}}^{+\frac{2\pi}{3\tau\sqrt{3}}}dq'\tilde{\psi}(q’)\left[\int_{-\frac{\pi}{3\tau\sqrt{3}}}^{+\frac{2\pi}{3\tau\sqrt{3}}}dq" \tilde{\chi}(q")\varrho(q'-q")\right]\\
       &=& \int_{-\frac{\pi}{3\tau\sqrt{3}}}^{+\frac{2\pi}{3\tau\sqrt{3}}}dudv \tilde{\psi}(u)\tilde{\chi}(v)\tilde{\varrho}(q-u-v)\\
       &=&  ((\tilde{\psi}*\tilde{\chi})* \tilde{\varrho})(q).
  \end{eqnarray*}
  The associativity can also be seen immediately from the convolution theorem:
      \begin{eqnarray*}
    \mathcal{F}_\tau\left[\left(\psi*_\tau \left( \chi *_\tau \varrho\right)\right)(x)\right](\xi)= \sqrt{\frac{\pi}{\tau\sqrt{3}}}\psi(\xi)\mathcal{F}_\tau[\left(\chi*_\tau \varrho\right)(x)](\xi)=\frac{\pi}{\tau\sqrt{3}}\psi(\xi)\chi(\xi)\varrho(\xi),
\end{eqnarray*}
and 
\begin{eqnarray*}
    \mathcal{F}_\tau[\left(\left(\psi*_\tau \chi\right)*_\tau \varrho\right)(x)](\xi)= \sqrt{\frac{\pi}{\tau\sqrt{3}}}\mathcal{F}_\tau[\left(\psi*_\tau \chi\right)(x)](\xi)\varrho(\xi)=\frac{\pi}{\tau\sqrt{3}}\psi(\xi)\chi(\xi)\varrho(\xi).
\end{eqnarray*}
\end{proof}
\item
\begin{proof}
    Let $\tilde{\psi},  \tilde{\chi}, \tilde{\varrho}\in \mathcal{S}\left(\mathbb{R}\right)$, we have
\begin{eqnarray*}
   \left(\left(\alpha\tilde{\psi}+ \beta \tilde{\chi}\right)*  \tilde{\varrho}\right)(q) &=&\int_{-\frac{\pi}{3\tau\sqrt{3}}}^{+\frac{2\pi}{3\tau\sqrt{3}}}dq'\left(\alpha\tilde{\psi}+ \beta \tilde{\varrho}\right)(q’)\tilde{\chi}(q-q')\\
   &= & \alpha\int_{-\frac{\pi}{3\tau\sqrt{3}}}^{+\frac{2\pi}{3\tau\sqrt{3}}}dq’\tilde{\psi}(q’)\tilde{\chi}(q-q') +\beta\int_{-\frac{\pi}{3\tau\sqrt{3}}}^{+\frac{2\pi}{3\tau\sqrt{3}}}dq’ \tilde{\varrho}(q’)\tilde{\chi}(q-q')\\
    &= & \alpha \left(\tilde{\psi}*  \tilde{\varrho}\right)(q)+\beta \left(\tilde{\chi}*  \tilde{\varrho}\right)(q)
\end{eqnarray*}
The deformed Fourier transform of this linearity is given by
\begin{eqnarray*}
    \mathcal{F}_\tau[\left(\left(\alpha\psi+ \beta\chi\right)*_\tau \varrho\right)(x)](\xi)= \alpha\mathcal{F}_\tau\left[\left(\left(\psi*_\tau \varrho\right)\right)(x)\right](\xi)+ \beta\mathcal{F}_\tau\left[\left(\left(\chi*_\tau \varrho\right)\right)(x)\right](\xi).
\end{eqnarray*}
\end{proof}
\item 
\begin{proof}: Let $\psi, \chi,\varrho\in \mathcal{S}\left(\mathbb{R}\right)$, we have
\begin{eqnarray*}
    \mathcal{F}_\tau[\left(\psi\chi\right)(x)](\xi)&=& \sqrt{\frac{\tau\sqrt{3}}{\pi}}  \int_{-\ell_{max}}^{+\ell_{max}} d\mu_\tau(x) \psi(x)\chi)(x)e^{-\frac{i}{\hbar}\xi \Theta_\tau(x)}\\
    &=& \frac{1}{2\pi\hbar}\int_{-\ell_{max}}^{+\ell_{max}} d\mu_\tau(x) \psi(x)\left[ \int_{-\infty}^{+\infty} d\eta \chi(\eta) e^{\frac{i}{\hbar}\eta\Theta_\tau(x)}   \right]\\
     &=&\frac{1}{2\hbar\sqrt{\pi\tau\sqrt{3}}}\int_{\mathbb{R}}d\eta \chi(\eta)\psi(\xi-\eta).
\end{eqnarray*}

\end{proof}

\item \begin{proof}
From our previous result such that
    \begin{eqnarray*}
    \mathcal{F}_\tau[\left(\psi*_\tau \chi\right)(x)](\xi)=\sqrt{\frac{\pi}{\tau\sqrt{3}}}\psi(\xi)\chi(\xi)=\sqrt{\frac{\pi}{\tau\sqrt{3}}}\mathcal{F}_\tau[\psi(x)](\xi)\mathcal{F}_\tau[\chi(x)](\xi)
\end{eqnarray*}
So its inverse Fourier transform is given by
 \begin{eqnarray*}
    \left(\psi*_\tau \chi\right)(x)=\sqrt{\frac{\pi}{\tau\sqrt{3}}}\psi(\xi)\chi(\xi)=\sqrt{\frac{\pi}{\tau\sqrt{3}}}\mathcal{F}_\tau^{-1}\left[\mathcal{F}_\tau[\psi](\xi)\mathcal{F}_\tau[\chi](\xi)\right].
\end{eqnarray*}
\end{proof}

\subsection*{Proof of the proposition \eqref{tz}  } \label{PRO}
\begin{proof}
    Let $\psi, \chi\in \mathcal{S}\left(\mathbb{R}\right)$, then  we have
    \begin{eqnarray*}
        \int_{\mathbb{R}}d\xi \overline{\psi(\xi)}\chi(\xi)&=&\sqrt{\frac{\tau\sqrt{3}}{\pi}}\int_{-\ell_{max}}^{+\ell_{max}}d\mu_\tau(x)\int_{\mathbb{R}}  \overline{\psi(\xi)}\chi(x) e^{-\frac{i}{\hbar}\xi\Theta_\tau(x)}   d\xi \\
        &=&\sqrt{\frac{\tau\sqrt{3}}{\pi}}\int_{-\ell_{max}}^{+\ell_{max}}d\mu_\tau(x)\overline{\int_{\mathbb{R}}  \psi(\xi)e^{\frac{i}{\hbar}\xi\Theta_\tau(x)} d\xi} \chi(x)\\
         &=&\sqrt{\frac{\tau\sqrt{3}}{2}}\times \hbar\sqrt{4\pi\tau\sqrt{3}} \int_{-\ell_{max}}^{+\ell_{max}}d\mu_\tau(x)\overline{\psi(x)}\chi(x)\\
         &=& 2\hbar\tau\sqrt{3}\int_{-\ell_{max}}^{+\ell_{max}}d\mu_\tau(x)\overline{\psi(x)}\chi(x).
    \end{eqnarray*}
    In particular
    \begin{eqnarray*}
        ||\mathcal{F}_\tau\psi||^2=2\hbar\tau\sqrt{3}||\psi||_{S_+}
    \end{eqnarray*}
\end{proof}

\end{enumerate}
\subsection*{Proof of the Lemma \eqref{LMA2}}\label{ba22}
\begin{proof}
	Equation (\ref{moment1}) is equivalent to 
	\begin{eqnarray*}
		i\frac{\tau\hbar \sqrt{3}}{2}\frac{d}{d\xi}= \left[\arctan\left(\frac{2\tau x-1}{\sqrt{3}}\right)
		+\frac{\pi}{6}\right]=
		\left[\arctan\left(\frac{2\tau x-1}{\sqrt{3}}\right)
		+\arctan\left(\frac{1}{\sqrt{3}}\right)\right].\label{eq}
	\end{eqnarray*} 
	From the following relation 
	\begin{eqnarray*}
		\arctan\alpha +\arctan \beta=\arctan \left(\frac{\alpha+\beta}{1-\alpha\beta}\right),\quad \mbox{with}\quad \alpha\beta<1,
	\end{eqnarray*}
	we deduce that
	\begin{eqnarray*}
		\tan \left[\arctan\left(\frac{2\tau x-1}{\sqrt{3}}\right)
		+\arctan\left(\frac{1}{\sqrt{3}}\right)\right]=\frac{\tau x \sqrt{3} }{2-\tau x}.
	\end{eqnarray*}
	Therefore, the position operator $\hat X$ is represented as follows 
	\begin{eqnarray*}
		\hat X&=&\frac{2}{\tau}\frac{\tan\left(i\frac{\tau\hbar \sqrt{3}}{2}\partial_\xi\right)}{\sqrt{3}+\tan\left(i\frac{\tau\hbar \sqrt{3}}{2}\partial_\xi\right)}\mathbb{I},\\
		X\psi(\xi)&=&\frac{2}{\tau}\frac{\tan\left(i\frac{\tau\hbar \sqrt{3}}{2}\partial_\xi\right)}{\sqrt{3}+\tan\left(i\frac{\tau\hbar \sqrt{3}}{2}\partial_\xi\right)}\psi(\xi).
	\end{eqnarray*}
	Using equation (\ref{r1}), the action of $\hat P$ on the  quasi-representation (\ref{F1}) reads as follows
	\begin{eqnarray*}
		\hat P\psi(x)&=& \frac{-i\hbar}{\hbar\sqrt{4\pi\tau\sqrt{3}}} \int_{-\infty}^{+\infty}d\xi \psi(\xi) D_x\left(  e^{i\frac{2\xi}{\tau \hbar \sqrt{3}}\left[\arctan\left(\frac{2\tau x-1}{\sqrt{3}}\right)
			+\frac{\pi}{6}\right]}\right)\cr
		&=& \frac{1}{\hbar\sqrt{4\pi\tau\sqrt{3}}} \int_{-\infty}^{+\infty}d\xi \xi\psi(\xi)  e^{i\frac{2\xi}{\tau \hbar \sqrt{3}}\left[\arctan\left(\frac{2\tau x-1}{\sqrt{3}}\right)
			+\frac{\pi}{6}\right]}.\label{q111}
	\end{eqnarray*}
	On the other hand, the action of $\hat P$ on the  representation (\ref{F1}) reads as follows
	\begin{eqnarray*}
		\hat P\psi(x)&=& \frac{1}{\hbar\sqrt{4\pi\tau\sqrt{3}}} \int_{-\infty}^{+\infty}\hat P \psi(\xi) d\xi e^{i\frac{2\xi}{\tau \hbar \sqrt{3}}\left[\arctan\left(\frac{2\tau x-1}{\sqrt{3}}\right)
			+\frac{\pi}{6}\right]}.\label{b5}
	\end{eqnarray*}
	By comparing equation (\ref{q111}) and equation (\ref{b5}), we obtain equation (\ref{qooo}) of Lemma 4.2	
	\begin{eqnarray*}
		\hat P \psi(\xi)= \xi\psi(\xi).
	\end{eqnarray*}	
\end{proof}

\end{document}